\documentclass[twocolumn,
showpacs,preprintnumbers,amsmath,amssymb]{revtex4}

\usepackage{amssymb,amsmath,amsthm}
\usepackage{graphicx}
\newtheorem{Thm}{Theorem}

\theoremstyle{definition}

\newcommand{\bra}[1]{{\left\langle #1 \right|}}
\newcommand{\ket}[1]{{\left| #1 \right\rangle}}

\newcommand{\T}{\mbox{$\mathrm{tr}$}}

\begin{document}
\title{Tsallis entropy, $q$-expectation and constraints on three-party quantum correlations}
\author{Jeong San Kim}
\email{freddie1@khu.ac.kr} \affiliation{
 Department of Applied Mathematics and Institute of Natural Sciences, Kyung Hee University, Yongin-si, Gyeonggi-do 446-701, Korea
}
\date{\today}

\begin{abstract}
We show that the mutually exclusive nature of classical and quantum correlations distributed in multi-party quantum systems can be characterized in terms of $q$-expectation. Using Tsallis-$q$ entropy and $q$-expectation, we first provide generalized definitions of classical and quantum correlations, and establish their trade-off relations in three-party quantum systems of arbitrary dimension with respect to $q$-expectation for $q\geq 1$. We also provide equivalence conditions for monogamy and polygamy inequalities of quantum entanglement and quantum discord distributed in three-party quantum systems of arbitrary dimension with respect to $q$-expectation for $q\geq 1$.
\end{abstract}

\maketitle

\section{Introduction}

One distinct property of quantum correlations from classical ones is
in their shareability among many parties. Quantum correlations have limited shareability and distribution in multi-party systems, whereas classical correlations
can be freely shared among parties. For example, quantum entanglement is known to
obey the {\em monogamy} property~\cite{T04, KGS}; if a pair of quantum systems $A$ and $B$ are in a maximally entangled state, they cannot be in an entangled state with any other system, namely, $C$. This monogamy of entanglement has been quantitatively characterized as monogamy inequalities in terms of various entanglement measures~\cite{ckw, ov, kds, KSRenyi, KimT, KSU}. We note that the entanglement of assistance shows {\em polygamy}(dual monogamy) property in multi-party quantum systems, which was also quantitatively characterized as polygamy inequalities~\cite{GBS, BGK, KimGP}

In fact, maximal entanglement between two systems $A$ and $B$ even prohibits them from sharing classical correlation with $C$, the third party. Moreover, a maximal classical correlation between two parties forbids sharing entanglement with other parties~\cite{KW}. Thus a perfect correlation between two systems $A$ and $B$ can even exclude the possibility of sharing different kind of correlation with other systems.  We also note that there are trade-off relations of other correlations such as Bell nonlocality~\cite{SG, CH} and quantum discord~\cite{SAPB, BZYW}.

{\em Tsallis entropy} is a one-parameter generalization of von Neumann entropy based on the concept of $q$-{\em expectation}, and it plays an important role in various places~\cite{tsallis, lv}. In quantum information theory, the concavity of Tsallis-$q$ entropy for $q>0$ guarantees the property of {\em entanglement monotone}, a key requirement to construct a faithful entanglement measure~\cite{vidal}. Some conditions on separability criteria of quantum states can also be found by using Tsallis entropy~\cite{ar,tlb,rc}.

In nonextensive statistical mechanics, the concept of $q$-expectation in the definition of Tsallis-$q$ entropy is known to be theoretically consistent with the relative-entropy principle(or minimum cross-entropy principle),
which excludes the possibility of using the ordinary
expectation value from nonextensive statistical mechanics~\cite{LP, Abe}.
There are also discussions about characterizing classical statistical correlations inherented in quantum states in terms of $q$-expectation and Tsallis-$q$ entropy~\cite{rr, bpcp}.

In this paper, inspired by the concept of $q$-expectation, we show that
the mutually exclusive nature among classical and quantum correlations in multi-party quantum systems can be characterized in terms of $q$-expectation for the full range of $q\geq 1$. Using Tsallis-$q$ entropy and $q$-expectation, we first provide generalized definitions of classical and quantum correlations such as one-way classical correlation, quantum entanglement and quantum discord as well as their dual quantities. By investigating their properties, we establish some classes of trade-off relations among classical and quantum correlations in three-party quantum systems of arbitrary dimension in terms of the generalized correlation measures. We also provide equivalence conditions for monogamy and polygamy inequalities of quantum entanglement and quantum discord distributed in three-party quantum systems of arbitrary dimension with respect to $q$-expectation. When $q$ tends to 1, Tsallis-$q$ entropy and $q$-expectation are reduced to von Neumann entropy and ordinary expectation, respectively. Thus, our results encapsulate previous results of trade-off relations as special cases.

This paper is organized as follows. In Sec.~\ref{subsec: Tentropy},
we first recall the concept of $q$-expectation in accordance with the definition of Tsallis-$q$ entropy, and provide generalized definitions of entanglement quantifications in terms of Tsallis-$q$ entropy and $q$-expectation.
In Sec.~\ref{subsec: gqcorrelations}, we generalize the definitions of various classical and quantum correlations in terms of Tsallis-$q$ entropy and $q$-expectation. In Sec.~\ref{sec: toff}, we establish some classes of trade-off relations among classical and quantum correlations distributed in three-party quantum systems in terms of $q$-expectation for $q\geq1$.
In Sec.~\ref{sec: monopoly}, we provide equivalence conditions of monogamy and polygamy inequalities for quantum entanglement and quantum discord in three-party quantum systems in terms of $q$-expectation for $q\geq 1$.
Finally, we summarize our results in Sec.~\ref{Conclusion}.

\section{$q$-expectation and quantum correlations}
\label{sec: q-exp}

\subsection{Tsallis entropy and $q$-expected entanglement}
\label{subsec: Tentropy}

For a quantum state $\rho$, its Tsallis-$q$ entropy is defined as
\begin{align}
S_{q}\left(\rho\right)=-\T \rho ^{q} \ln_{q} \rho,
\label{Qtsallis}
\end{align}
where
\begin{eqnarray}
\ln _{q} x &=&  \frac {x^{1-q}-1} {1-q},
\label{qlog}
\end{eqnarray}
is the {\em generalized logarithm} of the real parameter $q$ with $q\geq0$ and
$~q \ne 1$~\cite{tsallis,lv}. As the generalized logarithm in Eq.~(\ref{qlog}) converges to the natural logarithm when $q$ tends to 1,
\begin{align}
\lim_{q\rightarrow 1}\ln _{q} x =\ln x,
\label{limitlog}
\end{align}
the singularity at $q=1$ that arises in the definition of Tsallis-$q$ entropy can be replaced by von Neumann entropy,
\begin{equation}
\lim_{q\rightarrow 1}S_{q}\left(\rho\right)=-\T\rho \ln \rho=:S\left(\rho\right).
\end{equation}
Thus we can simply denote $S_{1}\left(\rho\right)=S\left(\rho\right)$.

For a quantum state $\rho$ with its spectral decomposition
\begin{align}
\rho=\sum_i \lambda_i \ket{e_i}\bra{e_i},
\label{specdec}
\end{align}
its Tsallis-$q$ entropy in Eq.~(\ref{Qtsallis}) can be rewritten as
\begin{align}
S_{q}\left(\rho\right)=-\sum_{i}\lambda_{i}^q \ln _{q}\lambda_i = \frac{1}{1-q}\left[\sum_{i}\lambda_{i}^q -1\right].
\label{Ctsallis}
\end{align}
In other words, Tsallis-$q$ entropy is defined as the $q$-{\em expectation} of the generalized logarithms of the spectrum. Thus the generalization of von Neumann entropy into Tsallis-$q$ entropy is based on $q$-expectation of nonnegative real parameter $q$.

For $q\geq0$, Tsallis-$q$ entropy is a concave function. We also note that Tsallis-$q$ entropy is pseudoadditive, that is,
\begin{align}
S_{q}\left(\rho\otimes \sigma\right)=S_{q}\left(\rho\right)+S_{q}\left(\sigma\right)+\left(1-q\right)S_{q}\left(\rho\right)S_{q}\left(\sigma\right)
\label{nonext}
\end{align}
for any quantum states $\rho$ and $\sigma$.

Inspired by the concept of $q$-expectation in accordance with the definition of Tsallis-$q$ entropy, here we provide a class of bipartite entanglement measures; for $q\geq0$ and a bipartite pure state $\ket{\psi}_{AB}$, we define its $q$-{\em expected entanglement}($q$-E) as
\begin{equation}
{E}_{q}\left(\ket{\psi}_{AB} \right)=S_{q}(\rho_A),
\label{qEpure}
\end{equation}
where $\rho_A=\T _{B} \ket{\psi}_{AB}\bra{\psi}$ is the reduced density matrix of $\rho_{AB}$ on subsystem $A$.
For a bipartite mixed state $\rho_{AB}$, we define its $q$-E as the minimum $q$-expectation
\begin{equation}
E_{q}\left(\rho_{AB} \right)=\min \sum_i p^q_i E_{q}(\ket{\psi_i}_{AB}),
\label{qEmixed}
\end{equation}
over all possible pure state
decompositions of $\rho_{AB}$,
\begin{equation}
\rho_{AB}=\sum_{i} p_i |\psi_i\rangle_{AB}\langle\psi_i|.
\label{decomp}
\end{equation}

When $q$ tends to 1, Tsallis-$q$ entropy converges to von Neumann entropy and the $q$-expectation becomes ordinary expectation, therefore
\begin{align}
\lim_{q\rightarrow1}E_{q}\left(\rho_{AB} \right)=E_{\rm f}\left(\rho_{AB} \right),
\end{align}
where
\begin{align}
E_{\rm f}\left(\rho_{AB}\right)&=\min\sum_{i}p_i E(\ket{\psi_i}_{AB})
\label{eof}
\end{align}
is the {\em entanglement of formation}(EoF) of $\rho_{AB}$~\cite{bdsw}.
Thus $q$-E is one-parameter generalization of EoF for the full range of nonnegative parameter $q$ based on $q$-expectation.

As a dual quantity to $q$-E, we define $q$-{\em expected entanglement of assistance}($q$-EoA),
\begin{equation}
E^a_{q}\left(\rho_{AB} \right)=\max \sum_i p_i^q E_{q}(\ket{\psi_i}_{AB}),
\label{qEoA}
\end{equation}
where the maximum is taken over all possible pure state
decompositions of $\rho_{AB}$.
Similarly, we have
\begin{align}
\lim_{q\rightarrow1}E^a_{q}\left(\rho_{AB}
\right)=E^a\left(\rho_{AB} \right),
\label{TsallistoEoA}
\end{align}
where $E^a(\rho_{AB})$ is the entanglement of assistance(EoA)
of $\rho_{AB}$ defined as~\cite{cohen}
\begin{equation}
E^a(\rho_{AB})=\max \sum_{i}p_i E(\ket{\psi_i}_{AB}).
\label{eoa}
\end{equation}

\subsection{Generalized quantum correlations in terms of $q$-expectations}
\label{subsec: gqcorrelations}
For an ensemble representation $\mathcal E = \{p_i, \rho_i\}$ of a quantum state $\rho$~\cite{pdecomp}, its {\em Tsallis-$q$ difference} is defined as~\cite{Kim16T}
\begin{align}
\chi_q\left(\mathcal E\right)=S_q\left(\rho\right)-\sum_{i}p_{i}^q S_q\left(\rho_i\right).
\label{eq: q-diff}
\end{align}
Due to the the concavity of Tsallis-$q$ entropy, Tsallis-$q$ difference is always nonnegative for $q\geq 1$, and it converges to the Holevo quantity,
\begin{align}
\chi\left(\mathcal E\right)=S\left(\rho\right)-\sum_{i}p_i S\left(\rho_i\right),
\label{eq: holevo}
\end{align}
when $q$ tends to $1$.

Now let us consider a bipartite quantum state $\rho_{AB}$ with its reduced density matrix $\rho_A=\T_A\rho_{AB}$. We note that each measurement $\{M^x_B\}$ applied on subsystem $B$ induces a probability ensemble $\mathcal E = \{p_x, \rho_A^x\}$
of $\rho_A$ in the way that $p_x\equiv \T[(I_A\otimes M_B^x)\rho_{AB}]$ is the probability of the outcome $x$ and
$\rho^x_A=\T_B[(I_A\otimes {M_B^x})\rho_{AB}]/p_x$ is the state
of system $A$ when the outcome was $x$. The {\em one-way classical correlation}(CC)
of a bipartite state $\rho_{AB}$ is then defined as the maximum
Holevo quantity of $\rho_A =\T_B(\rho_{AB})$ over all
probability ensembles $\mathcal E = \{p_x, \rho_A^x\}$ of $\rho_A$ induced by the measurement on $B$~\cite{KW},
\begin{align}
{\mathcal J}^{\leftarrow}(\rho_{AB})&= \max_{\mathcal E}
\chi\left(\mathcal E\right).
\label{CC}
\end{align}

As a dual quantity to CC, the {\em one-way unlocalizable entanglement}(UE)~\cite{BGK} of $\rho_{AB}$ is defined as
the minimum Holevo quantity of $\rho_A =\T_B(\rho_{AB})$
\begin{align}
{\mathbf u}E^{\leftarrow}(\rho_{AB}) &=  \min_{\mathcal E} \chi\left(\mathcal E\right)
\label{EU}
\end{align}
over all possible probability ensembles of $\rho_A$ induced by {\em rank-1} measurements on subsystem $B$.
(To avoid the trivial minimum that is always zero when each measurement operatorof subsystem $B$ is proportional to the identity operator, the definition of UE for $\rho_{AB}$ only considers possible rank-1 measurements of $B$.)

Using Tsallis-$q$ difference and $q$-expectation, we generalize CC in Eq.~(\ref{CC}) for any real parameter $q\geq0$; {\em one-way classical $q$-correlation}
($q$-CC) of a bipartite state $\rho_{AB}$ is defined as
\begin{align}
{\mathcal J}_q^{\leftarrow}(\rho_{AB})&= \max_{\mathcal E} \chi_q\left(\mathcal E\right)
\label{qCC}
\end{align}
where the maximum is taken over all ensemble representations $\mathcal E$ of $\rho_A$ induced by measurements on subsystem $B$. Similarly, UE in Eq.~(\ref{EU}) can also be generalized as the minimum Tsallis-$q$ difference
\begin{align}
{\mathbf u}E_q^{\leftarrow}(\rho_{AB}) &= \min_{\mathcal E} \chi_q\left(\mathcal E\right),
\label{qUE}
\end{align}
over all probability ensemble representations $\mathcal E $ of $\rho_A$ induced by rank-1 measurements on subsystem $B$.

The quantity in Eq.~(\ref{qUE}) is referred to as the {\em one-way unlocalizable $q$-entanglement}($q$-UE)~\cite{Kim16T}, which can be considered as a dual quantity to $q$-CC in Eq.~(\ref{qCC}). Moreover, the continuity of Tsallis-$q$ difference with respect to $q$ naturally leads us to
\begin{align}
\lim_{q\rightarrow1}{\mathcal J}_q^{\leftarrow}(\rho_{AB})={\mathcal J}^{\leftarrow}(\rho_{AB})
\label{contJq}
\end{align}
and
\begin{align}
\lim_{q\rightarrow1}{\mathbf u}E_q^{\leftarrow}(\rho_{AB})={\mathbf u}E^{\leftarrow}(\rho_{AB})
\label{contJq}
\end{align}
for any bipartite quantum state $\rho_{AB}$.

Besides generalized entanglements such as $q$-E, $q$-EoA  and $q$-UE, we note that
the concept of $q$-expectation also enables us to generalize a different kind of quantum correlation, namely, {\em quantum discord}~\cite{discord}.
For a bipartite quantum state $\rho_{AB}$, its quantum discord is defined as
\begin{align}
\delta^{\leftarrow}(\rho_{AB})={\mathcal I}\left(\rho_{AB}\right)-{\mathcal J}^{\leftarrow}(\rho_{AB}),
\label{dis}
\end{align}
which is the difference between quantum mutual information
\begin{align}
{\mathcal I}\left(\rho_{AB}\right)=S(\rho_A)+S(\rho_B)-S(\rho_{AB})
\label{qmut}
\end{align}
and CC of $\rho_{AB}$ in Eq.~(\ref{CC}).
Moreover, the duality between CC and UE provides us with a dual quantity to quantum discord; {\em one-way unlocalizable quantum discord}(UD) of $\rho_{AB}$ is defined as~\cite{XFL12}
\begin{align}
{\mathbf u}\delta^{\leftarrow}(\rho_{AB})={\mathcal I}\left(\rho_{AB}\right)-{\mathbf u}E^{\leftarrow}(\rho_{AB}).
\label{udis}
\end{align}

Quantum mutual information in Eq.~(\ref{qmut}) can be generalized in terms of Tsallis-$q$ entropy. For $q\geq0$ and a bipartite quantum state $\rho_{AB}$, its {\em Tsallis-$q$ mutual entropy} is defined as
\begin{align}
{\mathcal I}_q\left(\rho_{AB}\right)=S_q\left(\rho_A\right)+S_q\left(\rho_B\right)-
S_q\left(\rho_{AB}\right).
\label{eq: qmutul}
\end{align}
By using Tsallis-$q$ mutual entropy and $q$-CC in Eq.~(\ref{qCC}), quantum discord can be generalized as
\begin{align}
\delta_q^{\leftarrow}(\rho_{AB})={\mathcal I}_q\left(\rho_{AB}\right)-{\mathcal J}_q^{\leftarrow}(\rho_{AB}),
\label{qdis}
\end{align}
which is referred to as {\em quantum $q$-discord}($q$-D)\cite{qdiscord}.

To close this section, we provide a dual quantity to $q$-D. The {\em one-way unlocalizable quantum $q$-discord}($q$-UD) of a bipartite state $\rho_{AB}$ is defined as
\begin{align}
{\mathbf u}\delta_q^{\leftarrow}(\rho_{AB})={\mathcal I}_q\left(\rho_{AB}\right)-{\mathbf u}E_q^{\leftarrow}(\rho_{AB}),
\label{qudis}
\end{align}
where ${\mathbf u}E_q^{\leftarrow}(\rho_{AB})$ is the $q$-UE of $\rho_{AB}$ in Eq.~(\ref{qUE}).

\section{Trade-off relations}
\label{sec: toff}
In this section, we establish some classes of trade-off relations among classical and quantum correlations in terms of $q$-expectation. The following theorem says that $q$-CC and $q$-E
as well as $q$-UE and $q$-EoA are mutually exclusive in three-party quantum systems.
\begin{Thm}
For $q \geq 1$ and a three-party pure state $\ket{\psi}_{ABC}$ with its reduced density matrices $\rho_{AB}=\T_C\ket{\psi}_{ABC}\bra{\psi}$,
$\rho_{AC}=\T_B\ket{\psi}_{ABC}\bra{\psi}$ and $\rho_{A}=\T_{BC}\ket{\psi}_{ABC}\bra{\psi}$, we have
\begin{align}
S_q(\rho_A)={\mathcal J}_q^{\leftarrow}(\rho_{AB})+E_q\left(\rho_{AC}\right)
\label{qCCEq}
\end{align}
and
\begin{align}
S_q(\rho_A)={\mathbf u}E_q^{\leftarrow}(\rho_{AB})+E^a_q\left(\rho_{AC}\right).
\label{qUEEqa}
\end{align}
\label{thm: qUEEqa}
\end{Thm}
Before we prove the theorem, we first note that each rank-1 measurement $\{M_B^x\}$ applied on system $B$ of $\ket{\psi}_{ABC}$ induces a pure-state decomposition
of $\rho_{AC}$,
\begin{align}
\rho_{AC}=\sum_{x}p_x \ket{\phi^x}_{AC}\bra{\phi^x}
\label{ensemble}
\end{align}
in the way that
\begin{align}
p_x=&\T[(I_{A}\otimes M_B^x\otimes I_{C})\ket{\psi}_{ABC}\bra{\psi}]
\label{pxcor}
\end{align}
and
\begin{align}
\ket{\phi^x}_{AC}\bra{\phi^x}&= \T_B[(I_{A}\otimes M_B^x\otimes I_{C})\ket{\psi}_{ABC}\bra{\psi}]/p_x.
\label{purecor}
\end{align}
Moreover, it is also straightforward to verify that each pure-state decomposition of
$\rho_{AC}=\sum_{x}p_x \ket{\phi^x}_{AC}\bra{\phi^x}$
induces a rank-1 measurement $\{M_B^x\}$ of system $B$.

In other words, there is a one-to-one correspondence between the set of all rank-1 measurements on subsystem $B$ and the set of all pure-state decompositions of $\rho_{AC}$. Thus, for a given three-party pure state $\ket{\psi}_{ABC}$, any optimization over all pure state decompositions of $\rho_{AC}$ is equivalent to optimizing over all possible rank-1 measurements on subsystem $B$.
\begin{proof}
To prove Eq.~(\ref{qCCEq}), let us consider a rank-1 measurement $\{M_B^x\}$ inducing an optimal pure-state decomposition of
$\rho_{AC}=\sum_{x}p_x \ket{\phi^x}_{AC}\bra{\phi^x}$
realizing $E_q(\rho_{AC})$, that is,
\begin{align}
E_q(\rho_{AC})=&\sum_{x}p_x^qE_q\left(\ket{\phi^x}_{AC}\right)
=\sum_{x}p_x^q S_q\left(\rho^x_A\right)
\label{opEq}
\end{align}
with
\begin{align}
\rho^x_A=\T_C\ket{\phi^x}_{AC}\bra{\phi^x},~~\rho_A=\sum_{x}p_x\rho^x_A.
\label{Rax}
\end{align}

Because
\begin{align}
\rho^x_A=&\frac{1}{p_x}\T_{BC}[(I_{A}\otimes M_B^x\otimes I_{C})\ket{\psi}_{ABC}\bra{\psi}]\nonumber\\
=&\frac{1}{p_x}\T_{B}[(I_{A}\otimes M_B^x)\rho_{AB}]
\label{trBCB}
\end{align}
with
\begin{align}
p_x=&\T[(I_{A}\otimes M_B^x\otimes I_{C})\ket{\psi}_{ABC}\bra{\psi}]\nonumber\\
=&\T[(I_{A}\otimes M_B^x)\rho_{AB}],
\label{trBCBpro}
\end{align}
each $\rho^x_A$ can be obtained from $\rho_{AB}$ by measuring subsystem $B$ with respect to $M^x_B$. Thus we have
\begin{align}
S_q(\rho_A)-E_q(\rho_{AC})=&S_q(\rho_A)-\sum_{x}p_x^q S_q\left(\rho^x_A \right)\nonumber\\
\leq& {\mathcal J}_q^{\leftarrow}(\rho_{AB})
\label{qCCbig}
\end{align}
where the inequality is due to the definition of $q$-CC.

Conversely, let us assume an optimal measurement $\{M_B^x\}$ realizing ${\mathcal J}_q^{\leftarrow}(\rho_{AB})$, that is,
\begin{align}
{\mathcal J}_q^{\leftarrow}(\rho_{AB})= S_q(\rho_A)-\sum_{x}p_x^q S_q\left(\rho^x_A \right)
\label{CCq2}
\end{align}
with
\begin{align}
p_x=\T[(I_{A}\otimes M_B^x)\rho_{AB}]
\label{pxforJ}
\end{align}
and
\begin{align}
\rho^x_A=\T_B[(I_{A}\otimes M_B^x)\rho_{AB}]/{p_x}.
\label{rhoxforJ}
\end{align}
Although each operator $M_B^x$ may not be of rank-1 in general, we can take a decomposition of $M_B^x$
\begin{align}
M_B^x=\sum_{y}M_B^{xy}
\label{r1decomp}
\end{align}
into rank-1 non-negative operators $M_B^{xy}$, so that $\{ M_B^{xy}\}$ becomes a new rank-1 measurement of subsystem $B$.
Let
\begin{align}
p_{xy}=\T[(I_{A}\otimes M_B^{xy})\rho_{AB}]
\label{pxy}
\end{align}
and
\begin{align}
\rho^{xy}_A=\T_B[(I_{A}\otimes M_B^{xy})\rho_{AB}]/{p_{xy}},
\label{rhoxy}
\end{align}
so that
\begin{align}
p_x=\sum_{y}p_{xy},~~\rho_A^x=\sum_{y}\frac{p_{xy}}{p_x}\rho^{xy}_A.
\label{pxpxy}
\end{align}

Now, we have
\begin{widetext}
\begin{align}
S_q(\rho_A)-\sum_{xy}p_{xy}^q S_q\left(\rho^{xy}_A\right)=&S_q(\rho_A)-\sum_{x}p_{x}^q\sum_{y}\left(\frac{p_{xy}}{p_x}\right)^q S_q\left(\rho^{xy}_A\right)\nonumber\\
\geq &S_q(\rho_A)-\sum_{x}p_{x}^q\sum_{y}\frac{p_{xy}}{p_x} S_q\left(\rho^{xy}_A\right)\nonumber\\
\geq&S_q(\rho_A)-\sum_{x}p_{x}^q S_q\left(\sum_{y} \frac{p_{xy}}{p_x} \rho^{xy}_A\right)\nonumber\\
=&{\mathcal J}_q^{\leftarrow}(\rho_{AB})
\label{qCCsmall}
\end{align}
\end{widetext}
where the first inequality is due to the convexity of the function $x^q$ for $q\geq1$, the second inequality is from the concavity of Tsallis-$q$ entropy, and the last equality is from Eqs.~(\ref{CCq2}) and (\ref{pxpxy}).
Moreover, from the definition of $E_q(\rho_{AC})$, we also have
\begin{align}
\sum_{xy}p_{xy}^q S_q\left(\rho^{xy}_A\right)\geq E_q(\rho_{AC}),
\label{qesmall}
\end{align}
which, together with Inequality~(\ref{qCCsmall}), leads us to
\begin{align}
S_q(\rho_A)-E_q(\rho_{AC})\geq {\mathcal J}_q^{\leftarrow}(\rho_{AB}).
\label{qCCsmall2}
\end{align}
Now, Inequalities~(\ref{qCCbig}) and (\ref{qCCsmall2}) recovers Eq.~(\ref{qCCEq}).

To prove Eq.~(\ref{qUEEqa}), we note that the one-to-one correspondence between the set of all rank-1 measurements of subsystem $B$ and the set of all pure-states decompositions of $\rho_{AC}$ mentioned in Eqs.~(\ref{ensemble}), (\ref{pxcor}) and (\ref{purecor}) enables us to rewrite
the definition of $q$-UE in Eq.~(\ref{qUE}) as
\begin{align}
{\mathbf u}E_q^{\leftarrow}(\rho_{AB})= S_q(\rho_A)- \max\sum_x p^q_x S_q(\rho^x_A)
\label{UEq2}
\end{align}
where the maximum is taken over all possible pure-state decompositions $\rho_{AC}=\sum_{x}p_x \ket{\phi^x}_{AC}\bra{\phi^x}$ such that
$\rho^x_A=\T_C \ket{\phi^x}_{AC}\bra{\phi^x}$. From the definition of $q$-EoA in Eq.~(\ref{qEoA}), we have
\begin{align}
{\mathbf u}E_q^{\leftarrow}(\rho_{AB})= S_q(\rho_A)- E_q^a\left(\rho_{AC}\right),
\label{UEq3}
\end{align}
which recovers Eq.~(\ref{qUEEqa}).
\end{proof}

For $q=1$, Eqs.~(\ref{qCCEq}) and (\ref{qUEEqa}) recover
the trade-off relations in~\cite{KW} and \cite{BGK}, respectively.
Thus Theorem~\ref{thm: qUEEqa} encapsulates those results as special cases.

We also note that Eq.~(\ref{UEq3}) provides us with an extrinsic definition for $q$-UE of a bipartite quantum state $\rho_{AB}$ in relation with
its three-party purification $\ket{\psi}_{ABC}$. We first note that $S_q\left( \rho_A\right)=E_{q}\left(\ket{\psi}_{A(BC)}\right)$ represents the amount of entanglement of the pure state $\ket{\psi}_{ABC}$ with respect to the bipartition between $A$ and $BC$ quantified by Tsallis-$q$ entropy.
$E_q^a\left(\rho_{AC}\right)$ is the $q$-EoA of $\rho_{AC}$ representing the maximum average entanglement(with respect to $q$-expectation)
that is possible to be concentrated on the subsystem $AC$ with the assistance of $B$.
Thus ${\mathbf u}E_q^{\leftarrow}(\rho_{AB})$ is the residual entanglement that cannot be localized
on $AC$ by the local measurement of $B$. Thus the term {\em unlocalizable} naturally arises.

\begin{Thm}
For $q \geq 1$ and a three-party pure state $\ket{\psi}_{ABC}$,
we have
\begin{align}
S_q(\rho_A)={\mathbf u}\delta_q^{\leftarrow}(\rho_{BA})+{\mathbf u}E_q^{\leftarrow}(\rho_{CA}).
\label{qUEqD}
\end{align}
\label{thm: qUEqD}
\end{Thm}
\begin{proof}
From the definition of $q$-UD in Eq.~(\ref{qudis}) together with Eqs.~(\ref{eq: qmutul}) and (\ref{UEq3}), we have
\begin{align}
{\mathbf u}\delta_q^{\leftarrow}(\rho_{AB})
=&-S_q\left(\rho_{A|B}\right)+E_q^a\left(\rho_{AC}\right),
\label{qudtcon}
\end{align}
where
\begin{align}
S_q\left(\rho_{A|B}\right)=S_q\left(\rho_{AB}\right)-S_q\left(\rho_B\right)
\label{qcondent}
\end{align}
is the {\em Tsallis-$q$ conditional entropy} of $\rho_{AB}$.
Moreover, for a three-party pure state $\ket{\psi}_{ABC}$, Eqs.~(\ref{UEq3}) and (\ref{qudtcon}) are universal
with respect to subsystems, therefore
\begin{align}
{\mathbf u}E_q^{\leftarrow}(\rho_{CA})=&S_q\left(\rho_{C}\right)-E_q^a\left(\rho_{CB}\right),
\label{uniABC2}
\end{align}
and
\begin{align}
{\mathbf u}\delta_q^{\leftarrow}(\rho_{BA})=&-S_q\left(\rho_{B|A}\right)+E_q^a\left(\rho_{BC}\right).
\label{uniABC1}
\end{align}

From Eqs.~(\ref{uniABC2}) and (\ref{uniABC1}), we have
\begin{align}
{\mathbf u}\delta_q^{\leftarrow}(\rho_{BA})+ {\mathbf u}E_q^{\leftarrow}(\rho_{CA})=&-S_q\left(\rho_{B|A}\right)
+E_q^a\left(\rho_{BC}\right)\nonumber\\
&+S_q\left(\rho_{C}\right)-E_q^a\left(\rho_{CB}\right)\nonumber\\
=&S_q\left(\rho_A\right)
\label{qUEqD2}
\end{align}
where the second inequality is due to
\begin{align}
S_q\left(\rho_{C}\right)=S_q\left(\rho_{AB}\right)
\label{same12}
\end{align}
for a three-party pure state $\ket{\psi}_{ABC}$.
\end{proof}

For a three-party pure state $\ket{\psi}_{ABC}$, Theorem~\ref{thm: qUEqD} says that the total entanglement between $A$ and $BC$ quantified by Tsallis-$q$ entropy consists of quantum discord between $A$ and $B$ quantified by $q$-UD and the entanglement between $A$ and $C$ quantified by $q$-UE. Thus Theorem~\ref{thm: qUEqD} establishes a mutually exclusive nature between $q$-UE and $q$-UD distributed in three-party quantum systems.

\section{Equivalence in monogamy and polygamy inequalities of three-party $q$-expected correlations}
\label{sec: monopoly}
In this section, we show the equivalence of monogamy and polygamy inequalities for $q$-UE, $q$-EoA and $q$-UD distributed in three-party quantum systems. We first consider the relation between monogamy inequality of $q$-UE and polygamy inequality of $q$-EoA. The following theorem states that they are equivalent in three-party quantum systems.

\begin{Thm}
For $q\geq 1$, and
any three-party pure state $\ket{\psi}_{ABC}$,
monogamy inequality of $q$-UE is equivalent to polygamy inequality of $q$-EoA, that is,
\begin{align}
E^a_{q}\left(\ket{\psi}_{A(BC)}\right)
\leq& E^a_{q}\left(\rho_{AB}\right)+E^a_{q}\left(\rho_{AC}\right),
\label{Eqpoly3}
\end{align}
if and only if
\begin{align}
{\mathbf u}E_q^{\leftarrow}(\ket{\psi}_{A(BC)})\geq {\mathbf u}E_q^{\leftarrow}(\rho_{AB})+{\mathbf u}E_q^{\leftarrow}(\rho_{AC})
\label{UEmono3}
\end{align}
\label{thm: Eqsuff3}
\end{Thm}

\begin{proof}
From the definition of $q$-UE in Eq.~(\ref{qUE}), we have
\begin{align}
{\mathbf u}E_q^{\leftarrow}\left(\ket{\psi}_{A(BC)}\right)
=&\min_{\mathcal E} \chi_q\left(\mathcal E\right)\nonumber\\
=&\min \left[S_q (\rho_A)-\sum_x p_x^q S_q(\rho_A^x)\right]
\label{qUEeq}
\end{align}
where the minimization is over all possible ensembles $\mathcal E = \{p_x, \rho^x_A\}$ of $\rho_A$ induced by
rank-1 measurements on the composite subsystem $BC$.
Becaues $\ket{\psi}_{ABC}$ is a pure state, each rank-1 measurement of $BC$
induces a pure-state ensemble of $\rho_A$. Thus each $\rho^x_A$ in Eq.~(\ref{qUEeq}) is a pure state. So $S_q(\rho_A^x)=0$ for each $x$, and we have
\begin{align}
{\mathbf u}E_q^{\leftarrow}\left(\ket{\psi}_{A(BC)}\right)=S_q(\rho_A).
\label{TsUeeq}
\end{align}
Moreover, from the definition of $q$-EoA in Eq.~(\ref{qEoA}), we have
\begin{align}
E_q^a\left(\ket{\psi}_{A(BC)}\right)=S_q(\rho_A)
\label{EoAeq}
\end{align}
for any pure state $\ket{\psi}_{ABC}$.

Now we note that Eq.~(\ref{qUEEqa}) of Theorem~\ref{thm: qUEEqa}
is universal, that is,
\begin{align}
E^a_q\left(\rho_{AB}\right)=S_q(\rho_A)-{\mathbf u}E_q^{\leftarrow}(\rho_{AC})
\label{qUEEqa23a}
\end{align}
and
\begin{align}
E^a_q\left(\rho_{AC}\right)=S_q(\rho_A)-{\mathbf u}E_q^{\leftarrow}(\rho_{AB})
\label{qUEEqa23b}
\end{align}
for a three-party pure state $\ket{\psi}_{ABC}$. Thus Eqs.~(\ref{TsUeeq}), (\ref{EoAeq}),
(\ref{qUEEqa23a}) and (\ref{qUEEqa23b}) lead us to
\begin{align}
E^a_q&\left(\rho_{AB}\right)+E^a_q\left(\rho_{AC}\right)-E_q^a\left(\ket{\psi}_{A(BC)}\right)\nonumber\\
&={\mathbf u}E_q^{\leftarrow}\left(\ket{\psi}_{A(BC)}\right)-({\mathbf u}E_q^{\leftarrow}(\rho_{AB})+{\mathbf u}E_q^{\leftarrow}(\rho_{AC})),
\label{qEqUEineq1}
\end{align}
which completes the proof.
\end{proof}

For $q=1$, Inequality~(\ref{Eqpoly3}) is reduced to the polygamy inequality of EoA in three-party quantum systems,
which was shown to be true~\cite{BGK}. Thus Theorem~\ref{thm: Eqsuff3} provides us with a proof of Inequality (\ref{UEmono3}) for $q=1$, that is, the monogamy inequality of three-party entanglement in terms of UE. Moreover, Theorem~\ref{thm: Eqsuff3} also states that this equivalence is still valid for the $q$-expected correlations in the full range of $q \geq 1$.

Now, let us consider another equivalence of polygamy conditions in three-party quantum systems.
The following theorem shows that the polygamy inequality of $q$-EoA in (\ref{Eqpoly3}) even implies a polygamous
property of quantum discord, a different kind of quantum correlation.
\begin{Thm}
For $q\geq 1$, and any three-party pure state $\ket{\psi}_{ABC}$,
the polygamy inequality of $q$-EoA is equivalent to the polygamy inequality of $q$-UD, that is, Inequality~(\ref{Eqpoly3}) is true if and only if
\begin{align}
{\mathbf u}\delta_q^{\leftarrow}\left(\ket{\psi}_{A(BC)}\right)
\leq& {\mathbf u}\delta_q^{\leftarrow}(\rho_{AB})+{\mathbf u}\delta_q^{\leftarrow}(\rho_{AC}),
\label{Dqpoly3}
\end{align}
where ${\mathbf u}\delta_q^{\leftarrow}\left(\ket{\psi}_{A(BC)}\right)$ is the $q$-UD of the pure state $\ket{\psi}_{ABC}$ with respect to the bipartition between $A$ and $BC$
\label{thm: Dqpoly3}
\end{Thm}

\begin{proof}
For a bipartite pure state $\ket{\psi}_{AB}$ together with its environmental system $\ket{\phi}_C$,
the universality of Inequality~(\ref{qUEqD}) for $\ket{\psi}_{AB}\otimes \ket{\phi}_C$ implies
\begin{align}
S_q(\rho_B)=&{\mathbf u}\delta_q^{\leftarrow}(\rho_{AB})+{\mathbf u}E_q^{\leftarrow}(\rho_{CB})\nonumber\\
=&{\mathbf u}\delta_q^{\leftarrow}(\ket{\psi}_{AB})+{\mathbf u}E_q^{\leftarrow}\left(\ket{\phi}_C\bra{\phi}\otimes \rho_B\right),
\label{qUEqD2}
\end{align}
where $\rho_B=\T_A \ket{\psi}_{AB}\bra{\psi}$.

Because the subsystems $BC$ is in a product state, any measurement on subsystem $B$ leaves the subsystem $C$ intact, that is, $\ket{\phi}_C\bra{\phi}$. Thus, the definition of $q$-UE leads us to
\begin{align}
{\mathbf u}E_q^{\leftarrow}\left(\ket{\phi}_C\bra{\phi}\otimes \rho_B\right)
=&\min \left[S_q (\rho_C)-\sum_x p_x^q S_q(\rho_C^x)\right]\nonumber\\
=&\min \left[S_q (\ket{\phi}_C)-\sum_x p_x^q S_q(\ket{\phi}_C)\right]\nonumber\\
=&0
\label{qud0}
\end{align}
where the minimization is over all possible rank-1 measurement on subsystem $B$.
From Eqs.~(\ref{qUEqD2}) and (\ref{qud0}), we have
\begin{align}
{\mathbf u}\delta_q^{\leftarrow}(\ket{\psi}_{AB})=S_q(\rho_B),
\label{qUEqD3}
\end{align}
therefore
\begin{align}
{\mathbf u}\delta_q^{\leftarrow}(\ket{\psi}_{AB})=S_q(\rho_B)=S_q(\rho_A)={\mathbf u}\delta_q^{\leftarrow}(\ket{\psi}_{BA}),
\label{qUEqD3sam}
\end{align}
for any bipartite pure state $\ket{\psi}_{AB}$.

For a three-party pure state $\ket{\psi}_{ABC}$, the universality of Eq.~(\ref{qudtcon}) implies
\begin{align}
{\mathbf u}\delta_q^{\leftarrow}(\rho_{AB})+S_q\left(\rho_{A|B}\right)=&E_q^a\left(\rho_{AC}\right),\nonumber\\
{\mathbf u}\delta_q^{\leftarrow}(\rho_{AC})+S_q\left(\rho_{A|C}\right)=&E_q^a\left(\rho_{AB}\right).
\label{quDqEoA}
\end{align}
Because $\ket{\psi}_{ABC}$ is a pure state, we haves
\begin{align}
S_q(\rho_{AB})=S_q(\rho_C),~ S_q(\rho_{AC})=S_q(\rho_B),
\label{qentsame}
\end{align}
which lead to
\begin{align}
S_q\left(\rho_{A|B}\right)+S_q\left(\rho_{A|C}\right)=&S_q\left(\rho_{AB}\right)-S_q(\rho_B)\nonumber\\
&+S_q\left(\rho_{AC}\right)-S_q(\rho_C)\nonumber\\
=&0.
\label{conzero}
\end{align}
From Eqs.~(\ref{quDqEoA}) and (\ref{conzero}), we have
\begin{align}
{\mathbf u}\delta_q^{\leftarrow}(\rho_{AB})+{\mathbf u}\delta_q^{\leftarrow}(\rho_{AC})=&E_q^a\left(\rho_{AC}\right)+
E_q^a\left(\rho_{AB}\right).
\label{quDqEoAsam}
\end{align}

By considering $\ket{\psi}_{ABC}$ as a bipartite pure state with respect to the bipartition between $A$ and $BC$,
Eqs.~(\ref{EoAeq}) and (\ref{qUEqD3sam}) lead us to
\begin{align}
{\mathbf u}\delta_q^{\leftarrow}\left(\ket{\psi}_{A(BC)}\right)=E_q^a\left(\ket{\psi}_{A(BC)}\right).
\label{qudabc}
\end{align}
Now, Eqs.~(\ref{quDqEoAsam}) and (\ref{qudabc}) complete the proof.
\end{proof}

Because Inequality~(\ref{Eqpoly3}) is generally true for $q=1$~\cite{BGK},
Theorem~\ref{thm: Dqpoly3} proves Inequality (\ref{Dqpoly3}) for $q=1$, that is, the polygamy inequality of UD in three-party systems. Theorem~\ref{thm: Dqpoly3} also provides us with their equivalence for the full range of $q \geq 1$.


\section{Conclusion}
\label{Conclusion}

In this paper, we have shown that the mutually exclusive natures among classical and quantum correlations in multi-party quantum systems can be characterized in terms of $q$-expectation for the full range of $q\geq 1$. Using Tsallis-$q$ entropy and $q$-expectation, we have provided generalized definitions of classical and quantum correlations such that $q$-CC, $q$-E and $q$-D as well as their dual quantities $q$-UE $q$-EoA and $q$-UD. By investigating their properties, we have established some classes of trade-off relations in three-party quantum systems of arbitrary dimension with respect to $q$-expectation. We have also shown the equivalences between the monogamy inequality of $q$-UE and polygamy inequalities of $q$-EoA and $q$-UD distributed in three-party quantum systems with respect to $q$-expectation. As Tsallis-$q$ entropy and $q$-expectation are reduced to von Neumann entropy and ordinary expectation for $q=1$, our results encapsulate previous results of restricted shareability and trade-off relations of correlations as special cases.

The study of limited shareability and distribution of quantum correlations in multi-party
quantum systems is the key ingredient of many secure quantum communication protocols.
For example, the quantitative characterization of trade-off relations among classical and quantum correlations enables us to possibly quantify how much information an eavesdropper could potentially obtain about the secret key to be extracted in quantum cryptography. We also note that the study of higher-dimensional quantum system than not just qubits is preferred in many quantum information tasks; in quantum key distribution, the use of quantum states in higher dimensional systems increases coding density and provide stronger security compared to qubits.

Thus our results of trade-off relations among classical and quantum correlations in high dimensional quantum systems can be useful methods for the foundation of many secure quantum information and communication protocols. Noting the importance of the study on multi-party quantum correlations, our results here can also provide a rich reference for future work to understand the nature of multi-party quantum correlations.

\section*{Acknowledgments}
This research was supported by Basic Science Research Program through the National Research Foundation of Korea(NRF) funded by the Ministry of Education(NRF-2017R1D1A1B03034727).



\begin{thebibliography}{1}

\bibitem{T04}
B. M. Terhal, IBM J. Research and Development {\bf 48}, 71 (2004).

\bibitem{KGS}
J.~S.~Kim, G.~Gour and B.~ C.~ Sanders,
Contemp. Phys. {\bf 53}, 5 p. 417-432 (2012).

\bibitem{ckw}
V. Coffman, J. Kundu and W. K. Wootters, Phys. Rev. A {\bf 61},
052306 (2000).

\bibitem{ov}
T. Osborne and F. Verstraete, Phys. Rev. Lett. {\bf 96}, 220503
(2006).

\bibitem{kds}
J. S. Kim, A. Das and B. C. Sanders,
Phys. Rev. A {\bf 79}, 012329 (2009).

\bibitem{KSRenyi}
J. S. Kim and B. C. Sanders,
J. Phys. A: Math. and Theor. {\bf 43}, 445305 (2010).

\bibitem{KimT}
J. S. Kim,
Phys. Rev. A {\bf 81}, 062328 (2010).

\bibitem{KSU}
J. S. Kim and B. C. Sanders,
J. Phys. A: Math. and Theor. {\bf 44}, 295303 (2011).

\bibitem{GBS}
G. Gour, S. Bandyopadhay and B. C. Sanders, J. Math. Phys. {\bf
48}, 012108 (2007).

\bibitem{BGK}
F. Buscemi, G. Gour and J. S. Kim,
Phys. Rev. A {\bf 80}, 012324 (2009).

\bibitem{KimGP}
J. S. Kim,
Phys. Rev. A {\bf 85}, 062302 (2012).

\bibitem{KW}
M. Koashi and A. Winter, Phys. Rev. A {\bf 69}, 022309 (2004).

\bibitem{SG}
V. Scarani and N. Gisin, Phys. Rev. Lett. {\bf 87},
117901 (2001).

\bibitem{CH}
S. Cheng and M. J. W. Hall, Phys. Rev. Lett. Phys. Rev. Lett. {\bf 118},
010401 (2017).

\bibitem{SAPB}
A. Streltsov, G. Adesso, M. Piani, and D. Bruß, Phys. Rev. Lett. {\bf 109},
050503 (2012).

\bibitem{BZYW}
Y.-K. Bai, N. Zhang, M.-Y. Ye, and Z. D. Wang, Phys. Rev. A {\bf 88},
012123 (2013).

\bibitem{tsallis}
C. Tsallis, J. Stat. Phys. {\bf 52}, 479 (1988).

\bibitem{lv}
P. T. Landsberg and V. Vedral, Phys. Lett. A {\bf 247}, 211 (1998).

\bibitem{vidal}
G. Vidal, J. Mod. Opt. {\bf 47}, 355 (2000).

\bibitem{ar}
S. Abe and A. K. Rajagopal, Physica A {\bf 289}, 157 (2001).

\bibitem{tlb}
C. Tsallis, S. Lloyd and M. Baranger, Phys. Rev. A {\bf 63}, 042104 (2001).

\bibitem{rc}
R. Rossignoli and N. Canosa, Phys. Rev. A {\bf 66}, 042306 (2002).

\bibitem{LP}
A. R. Lima and T. J. P. Penna, Phys. Lett. A {\bf 256}  p.~221--226 (1999).

\bibitem{Abe}
S. Abe, Astrophys. Space Sci. {\bf 305} p.~241--245 (2006).

\bibitem{rr}
A. K. Rajagopal and R. W. Rendell, Phys. Rev. A {\bf 72}, 022322 (2005).

\bibitem{bpcp}
J. Batle, A. R. Plastino, M. Casas and A. Plastino, J. Phys. A {\bf 35}, 10311 (2002).

\bibitem{bdsw}
C. H. Bennett, D. P. DiVincenzo, J. A. Smolin and W. K. Wootters,
Phys. Rev. A {\bf 54}, 3824 (1996).

\bibitem{cohen}
O. Cohen, Phys. Rev. Lett. {\bf 80}, 2493 (1998).

\bibitem{pdecomp}
Eequivalently, a probability decomposition $\rho=\sum_{i}p_i\rho_i$.

\bibitem{Kim16T}
J. S. Kim,
Phys. Rev. A {\bf 94}, 062338 (2016).

\bibitem{discord}
H. Ollivier and W. H. Zurek,
Phys. Rev. Lett. {\bf 88}, 017109 (2001);
L. Henderson and V. Vedral,
J. Phys. A {\bf 34}, 6899 (2001).

\bibitem{XFL12}
Z. Xi, H. Fan and Y. Li,
Phys. Rev. A {\bf 85}, 052102 (2012).

\bibitem{qdiscord}
D. P. Chi, J. S. Kim and K. Lee,
Phys. Rev. A {\bf 87}, 062339 (2013).

\end{thebibliography}
\end{document}